\documentclass{article}
\usepackage[left=1in,top=1in,right=1in]{geometry}

\usepackage{amssymb, amsmath, amsthm, graphicx, hyperref, tikz}

\theoremstyle{plain}
\newtheorem{theorem}{Theorem}[section]
\newtheorem{lemma}[theorem]{Lemma}

\newtheorem{corollary}[theorem]{Corollary}

\newtheorem{conjecture}[theorem]{Conjecture}

\theoremstyle{definition}
\newtheorem{definition}[theorem]{Definition}

\theoremstyle{remark}

\counterwithin{equation}{section}
\newcommand{\np}{\operatorname{np}}
\newcommand{\p}{\operatorname{p}}

\title{Some results on Archdeacon's conjecture for rotation systems}

\author{Arahat Chikkatur\thanks{University of California Los Angeles, Los Angeles, California, USA. Email: {\tt arahatc@g.ucla.edu}.} \and Ji Zeng\thanks{Alfréd Rényi Institute of Mathematics, Budapest, Hungary. Supported by ERC Advanced Grants ``GeoScape'', no. 882971 and ``ERMiD'', no. 101054936. Email: {\tt zeng.ji@renyi.hu}.}}
\date{}

\begin{document}

\maketitle

\begin{abstract}
A rotation system on \(n\) elements assigns to each element a cyclic
order of the other \(n-1\) elements.  A four-element subset is
\emph{non-planar} if its induced rotation system cannot be realized by a
crossing-free drawing of \(K_4\). As a combinatorial strengthening
of Hill's conjecture on the crossing number of the complete graph, Archdeacon conjectured that every
rotation system on \(n\) elements has at least $H(n)=\frac{1}{4} \lfloor\frac {n}{2}\rfloor \lfloor\frac{n-1}{2}\rfloor \lfloor\frac{n-2}{2}\rfloor \lfloor\frac{n-3}{2}\rfloor$
non-planar four-element subsets.

We computationally verify Archdeacon's conjecture for $n\leq 10$ and
show that every extremal rotation system in these orders is realizable
by a simple drawing. With computer assistance, we prove that every rotation system on $n$
elements has at least $(8/9 - o(1)) H(n)$ non-planar four-element subsets.
We also present a proof by hand for a weaker lower bound of $(2/3-o(1)) H(n)$. Finally, extending recent work of
Felsner on antipodal pairs in drawings, we show that Archdeacon's conjecture holds for antipodally shellable rotation systems.
\end{abstract}

\section{Introduction}

A drawing of a graph in the plane represents its vertices by distinct
points and its edges by Jordan arcs joining their endpoints, with the
interior of an edge containing no vertex. We also assume any two edges cross properly at their internal intersections.
The crossing number \(\operatorname{cr}(G)\) is the minimum number of
crossings in a drawing of \(G\).  In the late 1950s Anthony Hill found
concrete drawings of the complete graph \(K_n\) having
\begin{equation}\label{eq:hill-number}
H(n)=\frac14
\left\lfloor\frac n2\right\rfloor
\left\lfloor\frac{n-1}2\right\rfloor
\left\lfloor\frac{n-2}2\right\rfloor
\left\lfloor\frac{n-3}2\right\rfloor
\end{equation}
crossings; the construction was published by Guy~\cite{Guy1960} and Harary--Hill~\cite{HararyHill1963}. Hill's conjecture asserts that $\operatorname{cr}(K_n)=H(n)$ for every \(n\). The conjecture is known to hold through order fourteen~\cite{Aichholzer2021}. For
arbitrary \(n\), Balogh, Lidick\'y, and Salazar~\cite{BaloghLidickySalazar2019} proved, using flag algebras, the asymptotic
lower bound
\[
\operatorname{cr}(K_n)>
  \bigl(0.98559895-o(1)\bigr)H(n).
\]

A \emph{rotation system} \(R=(p_v)_{v\in V}\) consists of a finite
ground set \(V\) and, for each \(v\in V\), a cyclic order \(p_v\) of
\(V\setminus\{v\}\).  A drawing \(D\) of \(K_n\) induces a rotation
system \(R(D)\): the order \(p_v\) is the clockwise cyclic order in
which the edges incident with \(v\) leave \(v\), recorded by their
other endpoints. A drawing is \emph{simple} (or \emph{good}) if
every two edges have at most one point in common, which is either a
common endpoint or a proper crossing.  Thus adjacent edges do not cross and two independent
edges cross at most once.  Standard local redrawing arguments show that
every crossing-minimal drawing must be simple.  We call a rotation system \emph{realizable} if it is
induced by a simple drawing.  Throughout the paper, rotation systems
are abstract combinatorial objects and need not be realizable.

\begin{figure}
\centering
\begin{tikzpicture}[x=.72cm,y=1cm,line cap=round,line join=round]
  \foreach \a/\b in {
    1/3,1/4,1/5,1/9,2/4,2/9,3/8,3/9,
    4/7,4/8,4/9,5/7,5/8,5/9,6/8,6/9}
    \draw[blue!65!black,thin] (\a,0)
      .. controls
        ({\a+(\b-\a)/3},{.04*(\b-\a)*(\b-\a)})
        and ({\a+2*(\b-\a)/3},{.04*(\b-\a)*(\b-\a)})
      .. (\b,0);
  \foreach \a/\b in {
    1/6,1/7,1/8,2/5,2/6,2/7,2/8,
    3/5,3/6,3/7,4/6,7/9}
    \draw[red!70!black,thin] (\a,0)
      .. controls
        ({\a+(\b-\a)/3},{-.04*(\b-\a)*(\b-\a)})
        and ({\a+2*(\b-\a)/3},{-.04*(\b-\a)*(\b-\a)})
      .. (\b,0);
  \draw[black,thin] (1,0)--(9,0);
  \foreach \v in {1,...,9}
    \node[circle,draw,fill=white,inner sep=1.15pt,font=\scriptsize]
      at (\v,0) {\v};
  \node[anchor=west,font=\scriptsize] at (10.65,0) {$
    \left(
    \begin{array}{cccccccc}
    2&6&7&8&9&5&4&3\\
    1&9&4&3&5&6&7&8\\
    1&9&8&4&5&6&7&2\\
    1&9&8&7&5&6&3&2\\
    1&9&8&7&6&2&3&4\\
    1&2&3&4&5&9&8&7\\
    1&2&3&6&5&4&8&9\\
    1&2&7&6&5&4&3&9\\
    1&7&8&6&5&4&3&2
    \end{array}
    \right)
  $};
\end{tikzpicture}
\caption{A drawing of \(K_9\) attaining Hill's number. The \(v\)-th row of the matrix
lists the clockwise rotation \(p_v\), begun with its least entry.}
\label{fig:hill-k9}
\end{figure}
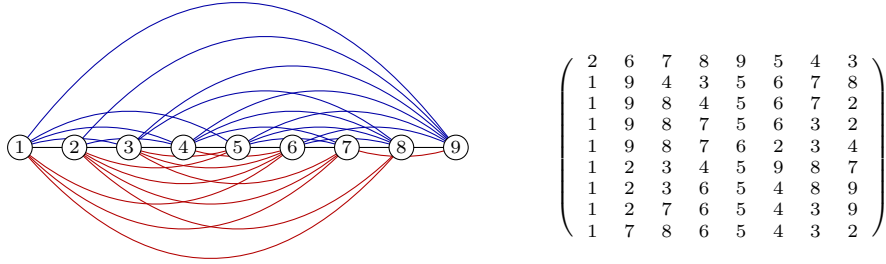

For \(U\subseteq V\), the subsystem \(R[U]\) induced by \(U\) is
obtained by restricting \(p_v\) to \(U\setminus\{v\}\) for each
\(v\in U\).  A four-element subset (four-subset for short) \(Q\subseteq V\) is \emph{planar}
if \(R[Q]\) is induced by a crossing-free drawing of \(K_4\), and is
\emph{non-planar} otherwise.  We write \(\np(R)\) for the number of
non-planar four-subsets of \(R\).

In a simple drawing of \(K_4\), the rotation system determines whether
the drawing is crossing-free or has one crossing.  Consequently, if
\(D\) is a simple drawing of \(K_n\) and its number of crossings denoted as \(\operatorname{cr}(D)\), we have
\begin{equation}\label{eq:crossings-equal-np}
\operatorname{cr}(D)=\np(R(D)).
\end{equation}

Dan Archdeacon proposed the following conjecture, which was recorded by Arroyo, McQuillan, Richter, and
Salazar~\cite{ArroyoMcQuillanRichterSalazar2017}.

\begin{conjecture}[Archdeacon]\label{archdeacon}
Every rotation system on \(n\) elements contains at least \(H(n)\)
non-planar four-element subsets.
\end{conjecture}

Since a crossing-minimal drawing is simple,
\eqref{eq:crossings-equal-np} shows immediately that
Conjecture~\ref{archdeacon} implies Hill's conjecture.
Archdeacon's conjecture is stronger: it makes the same assertion for
all rotation systems, including those that cannot be realized by a
simple drawing.  To our knowledge, no nontrivial general lower bound
for arbitrary rotation systems was previously known.
This paper is an initial attack on this conjecture.

Our first result is a computer verification of Archdeacon's conjecture
through order ten.

\begin{theorem}\label{thm:verification}
For \(n\leq 10\), every rotation system \(R\) on \(n\) elements
contains at least \(H(n)\) non-planar four-element subsets;  moreover,
if \(\np(R)=H(n)\), then \(R\) is realizable by a simple drawing.
\end{theorem}

Theorem~\ref{thm:verification} already gives a nontrivial asymptotic
consequence. Every rotation system $R$ on $n\geq 10$ elements satisfies \begin{equation*}
    \np(R)\geq
\frac{H(10)}{\binom{10}{4}}\binom {n}{4}
=\frac27\binom n4
=\left(0.76190476-o(1)\right)H(n).
\end{equation*}

Our second result further improves this asymptotic bound.
\begin{theorem}\label{thm:asymptotic}
Every rotation system on \(n\) elements contains at least $\left(8/9-o(1)\right)H(n)$ non-planar four-element subsets.
\end{theorem}

Our proof of this theorem relies on a certificate found by semidefinite
programming and verified by exhaustive computation.  We therefore also
give a second, weaker bound whose proof is entirely human-verifiable:
its only finite ingredient is the five-element case of
Theorem~\ref{thm:verification}, for which we give a short proof by hand
in Lemma~\ref{lem:five-hand}.
\begin{theorem}\label{thm:asymptotic_hand}
Every rotation system on \(n\) elements contains at least $(2/3 - o(1)) H(n)$ non-planar four-element subsets.
\end{theorem}

Hill's conjecture is known to hold for several structured classes of
drawings, including two-page, \(x\)-monotone, cylindrical,
\(x\)-bounded, shellable, bishellable, and spherical arc
drawings~\cite{AbregoEtAl2013,AbregoEtAl2014,AbregoEtAl2018,StreltsovaWagner2025}.
Most recently, Felsner introduced antipodally shellable drawings and
proved the conjecture for this class~\cite{Felsner2026}.

Two vertices \(u,v\) in a drawing are said to form an
\emph{antipodal pair} if no two edges incident with \(u\) or \(v\) cross.
Felsner proved that deleting such a pair has a particularly
well-behaved effect on the number of crossings.  We introduce the
corresponding notion for an abstract rotation system.  After cyclically
shifting \(p_u\), write \(p_u=(v,a_1,\ldots,a_{n-2})\) and let
\(p_u^v\) be the linear order \((a_1,\ldots,a_{n-2})\).
We call \(\{u,v\}\) an \emph{antipodal pair} of \(R\) if \(p_u^v\)
and \(p_v^u\) are reverses of each other.

\begin{theorem}\label{thm:antipodal}
If \(\{u,v\}\) is an antipodal pair in a rotation system \(R\) on
\(n\) elements, then
\[
\np(R)\geq \np(R\setminus\{u,v\})+H(n)-H(n-2).
\]
\end{theorem}

A rotation system \(R\) on \(n\) elements is \emph{antipodally
shellable} if its elements can be listed as \(v_1,\ldots,v_n\) so
that, for every \(0\leq i<\lfloor n/2\rfloor\), the pair
\(\{v_{n-2i-1},v_{n-2i}\}\) is antipodal in \(R[\{v_1,\ldots,v_{n-2i}\}]\). Clearly, Theorem~\ref{thm:antipodal} has the following consequence.

\begin{corollary}\label{cor:antipodal}
Conjecture~\ref{archdeacon} holds for antipodally shellable rotation
systems.
\end{corollary}

The paper is organized as follows.  Section~\ref{sec:verification}
sets up the combinatorial framework, in particular the orientation
notation used throughout, and describes the computation establishing
Theorem~\ref{thm:verification}.  Section~\ref{sec:asymptotic} proves
the two asymptotic bounds, Theorems~\ref{thm:asymptotic}
and~\ref{thm:asymptotic_hand}.  Section~\ref{sec:antipodal} proves
Theorem~\ref{thm:antipodal} and Corollary~\ref{cor:antipodal}.

\bigskip \noindent{\bf Acknowledgements.}
The mathematical content of this paper was initially produced by its human authors; we
used generative AI tools to implement the algorithms and to
prepare the manuscript. During manuscript preparation, ChatGPT-5.6 found a gap in the part of our proof of Theorem~\ref{thm:antipodal} that follows Lemma~\ref{lem:good-interval} and proposed a correct fix, which is its only contribution to the mathematical argument. We wish to thank Hui Tan for introducing the two
authors to each other in the context of a research experiences for
undergraduates (REU) project.

\section{Verification for small orders}\label{sec:verification}

For computational purposes, we identify the ground set of a rotation system with
\([n]=\{1,\ldots,n\}\), endowed with its natural order.  A cyclic order
\(p_i\) on \([n]\setminus\{i\}\) has a unique linear representative
whose first entry is the smallest element of \([n]\setminus\{i\}\); we
call it \emph{normalized}.
Writing these representatives as rows gives the \emph{matrix
representation}
\[
M(R)=(m_{i,j})_{\substack{1\leq i\leq n\\1\leq j\leq n-1}}
\]
of \(R\).  We compare rotation systems lexicographically after reading
their matrices row by row.

A permutation \(\sigma\in S_n\) acts on rotation systems by
relabelling.  Namely, if \(p_i=(x_1,\ldots,x_{n-1})\), then the rotation at \(\sigma(i)\) in \(\sigma(R)\) is the cyclic
order \((\sigma(x_1),\ldots,\sigma(x_{n-1}))\), normalized as above.  Two rotation systems are \emph{isomorphic} if
one is obtained from the other by a relabelling. The \emph{global reversal} \(R^{-1}\) is obtained by reversing every
cyclic order.  Thus, for a normalized row
\(p_i=(x_1,x_2,\ldots,x_{n-1})\), its reverse is
\[
p_i^{-1}=(x_1,x_{n-1},x_{n-2},\ldots,x_2).
\]
This operation corresponds geometrically to reflecting a drawing.  We
call \(R\) and \(R'\) \emph{equivalent} if \(R'\) is isomorphic to
either \(R\) or \(R^{-1}\). All classes mentioned below are equivalence
classes in this sense. We remark that much of the literature uses ``isomorphic'' for what we
call ``equivalent''; we keep the two apart as
in~\cite{BaloghLidickySalazar2019}, whose approach we follow closely
in the next section.

Let \(w(R)\) be the word obtained by reading \(M(R)\) row by row.  The
canonical representative \(\operatorname{can}(R)\) of the equivalence
class of \(R\) is characterized by
\[
w(\operatorname{can}(R))
=\min_{\substack{\sigma\in S_n\\ \epsilon\in\{1,-1\}}}
  w\bigl(\sigma(R^\epsilon)\bigr),
\]
where \(R^1=R\).  It can therefore be computed by examining
the \(2\cdot n!\) relabellings of \(R\) and its global reversal.

We shall use the following combinatorial criterion to decide whether a
given four-subset is planar.

\begin{definition}\label{def:orientation}
Let \(q\) be a cyclic order and let \(a,b,c\) be three distinct
elements of \(q\).  The \emph{orientation}
\(\varepsilon_q(a,b,c)\in\{0,1\}\) is defined by
\[
\varepsilon_q(a,b,c)=
\begin{cases}
0,&\text{if \(a,b,c\) occur in \(q\) in the cyclic order \((a,b,c)\)},\\
1,&\text{otherwise, that is, if they occur in the cyclic order \((a,c,b)\)}.
\end{cases}
\]
If \(q=p_v\) is the rotation at an element \(v\) of a rotation system,
we abbreviate \(\varepsilon_{p_v}\) to \(\varepsilon_v\).
\end{definition}

\begin{lemma}\label{planar}
Let \(R\) be a rotation system on a four-element set, and let
\((z_0,z_1,z_2,z_3)\) be any ordering of its ground set.  Then \(R\) is
planar if and only if
\begin{equation}\label{eq:planar-orientations}
\bigl(\varepsilon_{z_0}(z_1,z_2,z_3),\;
      \varepsilon_{z_1}(z_0,z_2,z_3),\;
      \varepsilon_{z_2}(z_0,z_1,z_3),\;
      \varepsilon_{z_3}(z_0,z_1,z_2)\bigr)
\in\bigl\{(0,1,0,1),\,(1,0,1,0)\bigr\}.
\end{equation}
In particular, whether this alternating pattern occurs does not depend
on which ordering of the ground set is used.
\end{lemma}

\begin{proof}
On the sphere a crossing-free drawing of \(K_4\) is the tetrahedral
embedding, unique up to reflection; its four faces are triangles, so we
may take \(z_0z_1z_2\) as the outer face and \(z_3\) inside.  Such a
drawing is determined by the clockwise order of \(z_0,z_1,z_2\) along
the outer triangle, and the two possibilities differ by interchanging
\(z_1\) and \(z_2\); the two correspond to the orientation patterns
\((1,0,1,0)\) and \((0,1,0,1)\).  Conversely, a cyclic order on three
elements is determined by its orientation, so the two alternating
patterns determine exactly these two systems, both of which are
planar.  The argument used no property of the ordering, so the
criterion does not depend on it.
\end{proof}

The next lemma is the five-element case of
Theorem~\ref{thm:verification}, without its ``moreover'' part.  Since
\(H(5)=1\), it is exactly the assertion that \(H(5)\) is a lower bound
at order five.  We isolate it here because, unlike the rest of
Theorem~\ref{thm:verification}, it admits a short proof by hand.

\begin{lemma}\label{lem:five-hand}
Every rotation system on five elements has at least one non-planar
four-subset.
\end{lemma}

\begin{proof}
Identify the ground set with \([5]\) and suppose, for contradiction,
that all five four-subsets of \(R\) are planar.  Relabelling
\([5]\setminus\{1\}\) along the cyclic order \(p_1\), we may assume
\(p_1=(2,3,4,5)\).  Then \(\varepsilon_1(a,b,c)=0\) for every triple
\(a<b<c\) in \(\{2,3,4,5\}\), so by Lemma~\ref{planar} each of the
four four-subsets containing \(1\), listed in increasing order, has
orientation vector \((0,1,0,1)\).

Write the remaining rows normalized as \(p_v=(1,x,y,z)\).  For
distinct \(a,b\in\{x,y,z\}\) we have \(\varepsilon_v(1,a,b)=0\)
exactly when \(a\) precedes \(b\) in \((x,y,z)\).  Reading the four
vectors \((0,1,0,1)\) at the elements \(2,3,4,5\) thus yields three
precedences at each, which chain into a total order and determine the
row:
\[
\begin{array}{c|c|c}
v&\text{precedences in }p_v&p_v\\ \hline
2&5\prec 4\prec 3&(1,5,4,3)\\
3&2\prec 5\prec 4&(1,2,5,4)\\
4&3\prec 2\prec 5&(1,3,2,5)\\
5&4\prec 3\prec 2&(1,4,3,2)
\end{array}
\]
Restricting these four rows to \(\{2,3,4,5\}\) gives
\(\varepsilon_2(3,4,5)=\varepsilon_3(2,4,5)=\varepsilon_4(2,3,5)
=\varepsilon_5(2,3,4)=1\), so the four-subset \(\{2,3,4,5\}\) has
orientation vector \((1,1,1,1)\).  This is not alternating, so
\(\{2,3,4,5\}\) is non-planar, a contradiction.
\end{proof}

The next criterion reduces realizability to a local condition.  It is
due to Kyn\v{c}l~\cite{Kyncl2020}; see also the classification of
small good drawings in~\cite{AbregoEtAl2015}.

\begin{lemma}\label{realizable}
Let \(R\) be a rotation system on at least five elements.  Then \(R\)
is realizable by a simple drawing if and only if \(R[U]\) is realizable
by a simple drawing for every \(U\in\binom{V(R)}5\).
\end{lemma}

In our implementation, each five-element subsystem is canonicalized
and compared with the catalogue of the five known realizable equivalence
classes on five elements, see, for example~\cite{AbregoEtAl2015}.

We record next a simple averaging identity.  Every non-planar
four-subset \(Q\) of \(R\) remains present in \(R\setminus\{v\}\)
precisely when \(v\notin Q\), so each is counted for exactly \(n-4\)
choices of \(v\), and double counting gives
\begin{equation}\label{averaging}
(n-4)\np(R)=\sum_{v\in V(R)}\np(R\setminus\{v\}).
\end{equation}
In particular some \(v\) satisfies
\(\np(R\setminus\{v\})\leq\frac{n-4}{n}\np(R)\).

We now describe how Theorem~\ref{thm:verification} is verified, taking
order nine as an example.  To enumerate representatives of the
systems of order nine with at most \(36\) non-planar four-subsets, we
enumerate representatives of the systems of order eight with at most
\(20\), extend each of them by one element in every possible way, and
keep or discard each extension according to its number of
non-planar four-subsets, computed by Lemma~\ref{planar}.  To enumerate the
representatives of order eight with at most \(20\), we extend those of
order seven with at most \(10\); to enumerate those, we extend the
representatives of order six with at most \(5\); and the
representatives of order six we enumerate directly.  The successive
cutoffs come from \eqref{averaging}: deleting a suitable element from
a system of order nine with at most \(36\) non-planar four-subsets
leaves at most \(\lfloor\frac59\cdot 36\rfloor=20\) of them, and
likewise \(\lfloor\frac48\cdot 20\rfloor=10\) and
\(\lfloor\frac37\cdot 10\rfloor=4\), the last being covered by the
slightly larger cutoff \(5\).  The recursion is exhaustive: a relabelling,
together with a global reversal if needed, that carries
\(R\setminus\{v\}\) to its canonical representative extends to \(R\)
by sending \(v\) to the new symbol, so \(R\) is equivalent to one of
the generated extensions.

The program {\tt extendRS.c}, included with the arXiv submission as an
ancillary file, implements this.  It reads a list of rotation systems
of a given order, generates every one-element extension, and outputs
one canonical representative of each equivalence class within a
prescribed cutoff.  It is written in the most straightforward manner,
at some cost in performance, but still suffices to verify
Conjecture~\ref{archdeacon} up to \(n=8\).  Extending a list is
trivially parallel, since separate threads may handle separate
systems; for the step from order eight to order nine we ran a lightly
parallelized version on \(66\) CPU cores, which took \(165\) CPU
hours.  It found exactly \(421\) equivalence classes of rotation
systems of order nine with at most \(36\) non-planar four-subsets.
All of them have exactly \(36\), and all pass the realizability test
of Lemma~\ref{realizable}; the count 421 agrees with Aichholzer's
independent enumeration of crossing-minimal
drawings~\cite{Aichholzer2021}.

Finally, \eqref{averaging} and Lemma~\ref{realizable} carry the
verification from an odd order to the next even one.  Let \(n\geq 5\)
be odd, suppose Theorem~\ref{thm:verification} holds at order \(n\),
and let \(R\) have order \(n+1\).  Since
\(H(n+1)=\frac{n+1}{n-3}H(n)\) for odd \(n\),
\[
(n-3)\np(R)=\sum_{v\in V(R)}\np(R\setminus\{v\})
\geq (n+1)H(n)=(n-3)H(n+1).
\]
If equality holds, then \(\np(R\setminus\{v\})=H(n)\) for every
\(v\), so every \(R\setminus\{v\}\) is realizable; as each
five-element subsystem of \(R\) avoids some \(v\),
Lemma~\ref{realizable} makes \(R\) realizable.  Taking \(n=9\) settles
order ten and completes the proof of Theorem~\ref{thm:verification}.

\section{Asymptotic lower bounds}\label{sec:asymptotic}

\subsection{Proof of Theorem~\ref{thm:asymptotic}}\label{subsec:flag}

Let \(R=(p_v)_{v\in V}\) be a rotation system on \(n\) elements.  For
a rotation system \(F\), let \(\p(F)\) be its number of planar
four-subsets, so that
\begin{equation}\label{eq:np-plus-p}
\np(F)+\p(F)=\binom{|V(F)|}{4}.
\end{equation}
In particular \(0\leq \p(F)\leq5\) when \(F\) has five elements.  The following lemma, whose proof is an obvious double-counting, expresses \(\p(R)\) in terms of these local counts.

\begin{lemma}\label{lem:k5-average}
For every rotation system \(R\) on \(n\geq 5\) elements,
\begin{equation}\label{eq:k5-planar-average}
\frac{1}{\binom n5}\sum_{S\in\binom V5}\p(R[S])
=\frac{5\,\p(R)}{\binom n4}.
\end{equation}
\end{lemma}

So we must bound the average of \(\p(R[S])\) from above.  Let
\(\mathcal R_5\) be the set of all \(6^5=7776\) labelled rotation
systems on \([5]\), and for \(F\in\mathcal R_5\) let \(x_F\) be the
number of five-element subsets \(U\subseteq V\) such that \(R[U]\),
with the elements of \(U\) relabelled by \([5]\) in increasing order,
equals \(F\), divided by \(\binom n5\).  Then \(x_F\geq0\) and \(\sum_Fx_F=1\).  In
other words, \(x=(x_F)\) is the distribution of the five-element
subsystems of \(R\).  Since \(\p\) does not change under relabelling,
\begin{equation}\label{eq:target-linear}
\frac{1}{\binom n5}\sum_{S\in\binom V5}\p(R[S])
=\sum_{F\in\mathcal R_5}\p(F)\,x_F .
\end{equation}

We now exploit the constraint on the vectors \(x\) that really come
from a rotation system.  Given a rotation system \(F\) on four
ordered elements \((a,b,c,d)\), we define its \emph{type vector} as
\[
\tau(F)=\bigl(\varepsilon_a(b,c,d),\varepsilon_b(a,c,d),
\varepsilon_c(a,b,d),\varepsilon_d(a,b,c)\bigr)\in\{0,1\}^4 .
\]
Listing all vectors in \(\{0,1\}^4\) as \(\tau_1,\ldots,\tau_{16}\) in
lexicographic order, Lemma~\ref{planar} says that \(F\) is planar if
and only if its type vector is \(\tau_6\) or \(\tau_{11}\).

\begin{definition}\label{def:five-vertex-flag-matrix}
Let \(F\) be a rotation system on a five-element set.  For
\(1\leq i,j\leq16\), let \(C_F(i,j)\) be the number of orderings
\((v_0,v_1,v_2,x,y)\) of its five elements such that the system
induced on the ordered set \((v_0,v_1,v_2,x)\) has type vector
\(\tau_i\) and the system induced on the ordered set
\((v_0,v_1,v_2,y)\) has type vector \(\tau_j\).  Write
\(M_F=\frac{1}{5!}C_F=\frac{1}{120}C_F\).
\end{definition}

The \((i,j)\) entry of \(\sum_Fx_FM_F\) is the probability that, for a
random ordered triple \(\theta\) of distinct elements of \(V\) and a
random ordered pair \((x,y)\) of distinct further elements, the
subsystems induced on \(\theta\cup\{x\}\) and \(\theta\cup\{y\}\) have
type vectors \(\tau_i\) and \(\tau_j\).  This matrix is almost
positive semidefinite.  Indeed, let \(z_\theta\) be the distribution
of the type vector of one random extension element.  Had \(x\) and
\(y\) been drawn independently, the matrix would be
\(z_\theta z_\theta^{\mathsf T}\), which is positive semidefinite, and
averaging over \(\theta\) keeps it so.  Asking for \(x\neq y\) changes
it by \(O(1/n)\), as the proof of Theorem~\ref{thm:asymptotic} below
shows.  Heuristically, in the limit we have
\begin{equation}\label{eq:psd-constraint}
\sum_{F\in\mathcal R_5}x_FM_F\succeq0 .
\end{equation}
So the quantity we want is, asymptotically, at most the value of the
semidefinite program
\begin{equation}\label{eq:sdp-primal}
\begin{aligned}
\text{maximize}\quad & \sum_{F\in\mathcal R_5}\p(F)\,x_F\\
\text{subject to}\quad & x_F\geq0,\qquad
\sum_{F\in\mathcal R_5}x_F=1,\qquad
\sum_{F\in\mathcal R_5}x_FM_F\succeq0 .
\end{aligned}
\end{equation}

The duality principle suggests us to look for a symmetric \(16\times16\) matrix
\(Q\succeq0\) and a constant \(A\) with
\begin{equation}\label{eq:dual-feasibility}
\p(F)+\langle Q,M_F\rangle\leq A
\qquad\text{for every }F\in\mathcal R_5 ,
\end{equation}
where \(\langle X,Y\rangle=\operatorname{tr}(X^{\mathsf T}Y)\). Such
a pair bounds \eqref{eq:sdp-primal} by \(A\).  Indeed, the inner product of two positive
semidefinite matrices is nonnegative, so for every feasible \(x\) we
have \(\sum_Fx_F\langle Q,M_F\rangle
=\langle Q,\sum_Fx_FM_F\rangle\geq0\), and
averaging \eqref{eq:dual-feasibility} with the weights \(x_F\), which
are nonnegative and sum to \(1\), gives
\[
\sum_F\p(F)x_F
\leq\sum_Fx_F\bigl(\p(F)+\langle Q,M_F\rangle\bigr)\leq A .
\]
We check \eqref{eq:dual-feasibility} by one inequality for each of
the \(7776\) systems in \(\mathcal R_5\). We found such a pair of \(Q\) and \(A\) numerically, by minimizing \(A\) subject to \eqref{eq:dual-feasibility} with the SDP solver Clarabel~\cite{GoulartChen2026}.
The numerical solution suggested \(A = 10/3\) and a rank-one matrix \(Q\); rescaling and rounding led to the following exact values
\[
s=(-1,0,0,1,0,1,-1,0,0,1,-1,0,-1,0,0,1)^{\mathsf T}\in\mathbb Z^{16},
\qquad
Q=\frac{10}{3}ss^{\mathsf T},
\qquad
A=\frac{10}{3} .
\]
This \(Q\) is positive semidefinite, being a positive multiple of
\(ss^{\mathsf T}\).  For this pair, \eqref{eq:dual-feasibility} becomes the
following statement, the only computer-assisted ingredient of the
proof.

\begin{lemma}
\label{lem:finite-k5-certificate}
For every rotation system \(F\) on five elements,
\begin{equation}\label{eq:finite-k5-certificate}
\p(F)+\langle Q,M_F\rangle\leq\frac{10}{3} .
\end{equation}
\end{lemma}

\begin{proof}
Since \(\langle Q,M_F\rangle
=\frac{10}{3}\cdot\frac{s^{\mathsf T}C_Fs}{120}
=\frac{s^{\mathsf T}C_Fs}{36}\), multiplying
\eqref{eq:finite-k5-certificate} by \(36\) turns it into the
inequality \(36\p(F)+s^{\mathsf T}C_Fs\leq120\) between integers, which
we check for every \(F\in\mathcal R_5\).  The left-hand side has maximum
\(120\) over \(\mathcal R_5\), so there is no violation, and as every quantity is an integer
the check is exact. This check is performed by the program {\tt veriFIVE.c}, included
with the arXiv submission as an ancillary file.
\end{proof}

Both the semidefinite program and the exhaustive check in Lemma~\ref{lem:finite-k5-certificate} are easily handled on a personal computer.

\begin{proof}[Proof of Theorem~\ref{thm:asymptotic}]
Let \(\theta=(v_0,v_1,v_2)\) be an ordered triple of distinct elements
of \(V\) and put \(N=n-3\).  For \(1\leq i\leq16\) let
\(c_i(\theta)\) be the number of \(x\notin\theta\) for which the
system induced on \((v_0,v_1,v_2,x)\) has type vector \(\tau_i\), so
\(\sum_ic_i(\theta)=N\).  Pick an ordered pair \((x,y)\) of distinct
elements outside \(\theta\) at random and let \(D_\theta(i,j)\) be the
probability that they give type vectors \(\tau_i\) and \(\tau_j\).
Counting ordered pairs,
\[
D_\theta(i,j)=
\frac{c_i(\theta)c_j(\theta)-\delta_{ij}c_i(\theta)}{N(N-1)},
\]
and hence
\begin{align}
\langle Q,D_\theta\rangle
&=\frac{10}{3N(N-1)}
\left[\left(\sum_{i=1}^{16}s_ic_i(\theta)\right)^2
-\sum_{i=1}^{16}s_i^2c_i(\theta)\right]\notag\\
&\geq-\frac{10}{3N(N-1)}\sum_{i=1}^{16}c_i(\theta)
=-\frac{10}{3(n-4)},
\label{eq:distinct-extension-bound}
\end{align}
where we dropped the square and used \(s_i^2\leq1\) and
\(N-1=n-4\).

Averaging \(D_\theta\) over all ordered triples of distinct elements
is the same as picking an ordered five-tuple of distinct elements at
random, which one may also do by picking a five-element subset \(S\)
and then one of the \(120\) orderings of \(S\).  Hence
\begin{equation}\label{eq:flag-matrix-average}
\frac{1}{n(n-1)(n-2)}\sum_{\theta}D_\theta
=\frac{1}{\binom n5}\sum_{S\in\binom V5}M_{R[S]}
=\sum_{F\in\mathcal R_5}x_FM_F ,
\end{equation}
and taking the inner product with \(Q\) in
\eqref{eq:distinct-extension-bound} termwise gives
\begin{equation}\label{eq:average-flag-lower-bound}
\Bigl\langle Q,\sum_{F\in\mathcal R_5}x_FM_F\Bigr\rangle
\geq-\frac{10}{3(n-4)} .
\end{equation}
Averaging Lemma~\ref{lem:finite-k5-certificate} with the weights
\(x_F\) and using \eqref{eq:average-flag-lower-bound},
\[
\sum_{F\in\mathcal R_5}\p(F)x_F
=\sum_{F\in\mathcal R_5}x_F
\left(\p(F)+\langle Q,M_F\rangle\right)
-\Bigl\langle Q,\sum_{F\in\mathcal R_5}x_FM_F\Bigr\rangle
\leq\frac{10}3+\frac{10}{3(n-4)} .
\]
With \eqref{eq:target-linear} and Lemma~\ref{lem:k5-average}, this
reads \(5\,\p(R)/\binom n4\leq\frac{10}3+\frac{10}{3(n-4)}\), so by
\eqref{eq:np-plus-p}
\[
\np(R)=\binom n4-\p(R)
\geq\left(\frac13-\frac{2}{3(n-4)}\right)\binom n4 .
\]
Finally, \eqref{eq:hill-number} gives
\(H(n)=\bigl(\frac38+O(1/n)\bigr)\binom n4\), so
\(\np(R)\geq\bigl(\frac89-O(1/n)\bigr)H(n)\), as wanted.
\end{proof}

\subsection{Proof of Theorem~\ref{thm:asymptotic_hand}}\label{subsec:hand}

This proof uses no computation beyond Lemma~\ref{lem:five-hand}.  With
Lemma~\ref{lem:k5-average}, that lemma alone already gives
\(\np(R)\geq\frac15\binom n4\).  A second moment estimate raises the
density \(1/5\) to \(1/4\).

\begin{proof}[Proof of Theorem~\ref{thm:asymptotic_hand}]
Let \(R\) be a rotation system on \(V\) with \(|V|=n\geq5\), and for
\(S\subseteq V\) abbreviate \(\np(R[S])\) and \(\p(R[S])\) to
\(\np(S)\) and \(\p(S)\).  Let \(B\) be the set of non-planar
four-subsets of \(R\), so that \(\np(R)=|B|\), and put
\(\beta=|B|/\binom n4\in[0,1]\).  Let \(\mathbb E\) denote expectation
over a uniformly random \(S\in\binom V5\).  Lemma~\ref{lem:five-hand}
gives
\begin{equation}\label{eq:np-range}
1\leq\np(S)\leq5
\qquad\text{for every }S\in\binom V5 ,
\end{equation}
while Lemma~\ref{lem:k5-average} gives
\(\mathbb E\bigl(\p(S)\bigr)=5\,\p(R)/\binom n4=5(1-\beta)\), so by
\eqref{eq:np-plus-p},
\begin{equation}\label{eq:first-moment}
\mathbb E\bigl(\np(S)\bigr)=5\beta .
\end{equation}

We next bound the second moment from below.  Two distinct four-subsets
lie in a common five-element set exactly when they meet in three
elements, and then their union is the only such set.  Writing
\(d(T)=|\{Q\in B: T\subseteq Q\}|\) for \(T\in\binom V3\), and
grouping the pairs \(\{Q,Q'\}\subseteq B\) with \(|Q\cap Q'|=3\) by
their intersection \(T=Q\cap Q'\),
\[
\sum_{S\in\binom V5}\binom{\np(S)}2
=\bigl|\{\{Q,Q'\}\subseteq B: |Q\cap Q'|=3\}\bigr|
=\sum_{T\in\binom V3}\binom{d(T)}2 .
\]
Counting the pairs \((T,Q)\) with \(T\subseteq Q\in B\) in two ways
gives \(\sum_Td(T)=4|B|\), so the average of \(d(T)\) is
\(4|B|/\binom n3=\beta(n-3)\).  By Jensen's inequality,
\[
\sum_{T\in\binom V3}\binom{d(T)}2
\geq\binom n3\binom{\beta(n-3)}2.
\]
Dividing by \(\binom n5\), we can compute
\begin{equation}\label{eq:second-moment}
\mathbb E\left(\binom{\np(S)}2\right)
\geq\frac{10\beta^2(n-3)-10\beta}{n-4}
\geq10\beta^2-\frac{10}{n-4},
\end{equation}
the last step using \(n-3\geq n-4\) and \(\beta\leq1\).

On \([1,5]\), convexity puts \(\binom x2\) below the chord through
\((1,0)\) and \((5,10)\), that is, \(x\geq1+\frac25\binom x2\).  By
\eqref{eq:np-range} we may put \(x=\np(S)\), take expectations and use
\eqref{eq:second-moment}:
\[
\mathbb E\bigl(\np(S)\bigr)
\geq1+\frac25\mathbb E\left(\binom{\np(S)}2\right)
\geq1+4\beta^2-\frac4{n-4} .
\]
Comparing with \eqref{eq:first-moment} gives
\(5\beta\geq1+4\beta^2-\frac4{n-4}\), that is,
\begin{equation}\label{eq:beta-quadratic}
(1-4\beta)(1-\beta)=4\beta^2-5\beta+1\leq\frac4{n-4} .
\end{equation}
If \(\beta<1/4\) then \(1-\beta>3/4\), so \eqref{eq:beta-quadratic}
gives \(1-4\beta\leq\frac{16}{3(n-4)}\), and otherwise
\(\beta\geq1/4\); in either case
\(\beta\geq\frac14-\frac4{3(n-4)}\), so
\(\np(R)\geq(\frac14-o(1))\binom n4\). Therefore, we can use \eqref{eq:hill-number} to conclude that
\[
\np(R)\geq\left(\frac14-o(1)\right)\cdot\frac83H(n)
=\left(\frac23-o(1)\right)H(n). \qedhere
\]
\end{proof}

\section{Antipodal pairs}\label{sec:antipodal}

\begin{proof}[Proof of Theorem~\ref{thm:antipodal}]
Let \(\{u,v\}\) be the given antipodal pair, and
\(U\) be the set of elements not equal to \(u\) or \(v\). Write \(m=|U|=n-2\). We identify \(U\)
with \([m]\) so that
\[
p_u=(v,1,\ldots,m)
\quad\text{and}\quad
p_v=(u,m,\ldots,1).
\]
Introduce an auxiliary symbol \(\star\), and let \(c_u^v\) be obtained
from \(p_u\) by replacing \(v\) with \(\star\), and \(c_v^u\) from
\(p_v\) by replacing \(u\) with \(\star\). By our hypothesis, the two cyclic orders \(c_u^v\) and \(c_v^u\) are reverses of each other.

We call a triple \(T\in\binom{U\cup\{\star\}}3\) \emph{good} if
\(T\cup\{u\}\) and \(T\cup\{v\}\) are planar in case
\(\star\notin T\), and if \(\{u,v\}\cup(T\setminus\{\star\})\) is
planar in case \(\star\in T\). For distinct \(a,b\in U\), deleting \(a\) and \(b\) from \(c_u^v\) gives us two intervals of elements; let \(I_{a,b}\) be the one containing \(\star\) and
\(J_{a,b}\) the other. Note that since \(c_v^u\) is the reverse of \(c_u^v\), deleting \(a\) and \(b\) from \(c_v^u\) gives us the same two intervals. An element \(x \in (U\cup\{\star\}) \setminus \{a,b\}\) is said to be a \emph{good complement} of \(\{a,b\}\) if the triple \(\{a,b,x\}\) is good.

\begin{lemma}\label{lem:good-interval}
For distinct \(a,b\in U\), at most one of the intervals \(I_{a,b}\) and \(J_{a,b}\) contains a good complement of \(\{a,b\}\).
\end{lemma}

\begin{proof}

Suppose on the contrary that \(x \in I_{a,b}\) and \(y \in J_{a,b}\) are two good complements of \(\{a,b\}\). We first deal with the case \(x \neq \star\). By our definition of \(I_{a,b}\) and \(J_{a,b}\), we have \(\varepsilon_u(a,b,x) = 1 - \varepsilon_u(a,b,y)\). Without loss of generality, we assume \(\varepsilon_u(a,b,x) = 0\) and \(\varepsilon_u(a,b,y) = 1\). Note that we also have \(\varepsilon_v(a,b,x) = 1 -\varepsilon_u(a,b,x)\) and similarly for \(y\). Because \(\{u,x,a,b\}\) and \(\{v,x,a,b\}\) are planar as \(x\) is a good complement, by Lemma~\ref{planar} applied to \((u,a,b,x)\) and \((v,a,b,x)\), we have \(\varepsilon_a(u,b,x)=1\) and \(\varepsilon_a(v,b,x)=0\). Similarly, because \(y\) is a good complement, we have \(\varepsilon_a(u,b,y)=0\) and \(\varepsilon_a(v,b,y)=1\).

Now consider the sequence of entries we encounter as we read the cyclic order \(p_a\) starting after and ending before \(b\): \(\varepsilon_a(u,b,x)=1\) means we encounter \(u\) before \(x\), \(\varepsilon_a(v,b,x)=0\) means we encounter \(x\) before \(v\), and together they imply that we encounter \(u\) before \(v\);  \(\varepsilon_a(u,b,y)=0\) means we encounter \(y\) before \(u\), \(\varepsilon_a(v,b,y)=1\) means we encounter \(v\) before \(y\), and together they imply that we encounter \(v\) before \(u\); hence a contradiction is reached.

It remains to deal with the case \(x=\star\), in which
\(\{u,v,a,b\}\) is planar.  Since \(\star\) takes the place of \(v\) in
\(c_u^v\), our definition of \(I_{a,b}\) and \(J_{a,b}\) gives
\(\varepsilon_u(a,b,v)=1-\varepsilon_u(a,b,y)\). Without loss of generality, we assume \(\varepsilon_u(a,b,v)=0\) and \(\varepsilon_u(a,b,y)=1\). Just as the previous case, we apply Lemma~\ref{planar} to \((u,a,b,v)\), \((u,a,b,y)\), and \((v,a,b,y)\), we can get \(\varepsilon_a(u,b,v)=1\), \(\varepsilon_a(u,b,y)=0\), and \(\varepsilon_a(v,b,y)=1\). Now we read the entries in \(p_a\) with \(b\) being the cut: \(\varepsilon_a(u,b,v)=1\) means we encounter \(u\) before \(v\);
\(\varepsilon_a(v,b,y)=1\) means we encounter \(v\) before \(y\), \(\varepsilon_a(u,b,y)=0\) means we encounter \(y\) before \(u\), and together they imply that we encounter \(v\) before \(u\); hence a contradiction is reached.
\end{proof}

Now we define a tournament \(\mathcal T\), i.e. a complete directed graph, on \([m]\) by setting, for
\(a<b\), \(a\to b\) if every good complement of \(\{a,b\}\) lies in
\(I_{a,b}\), and \(b\to a\) otherwise; by
Lemma~\ref{lem:good-interval} the second case means that every good
complement lies in \(J_{a,b}\).

Since \(\star \in I_{a,b}\) by definition, each edge
\(b\to a\) with \(a<b\) means that \(\star\) is not a good complement
of \(\{a,b\}\), and hence that \(\{u,v,a,b\}\) is non-planar. Now let \(a<b<c\) and suppose \(\{a,b,c\}\) is good.  Then we have
\(c\in I_{a,b}\), \(a\in I_{b,c}\), and \(b\in J_{a,c}\);
hence \(a\to b\to c\to a\).  So a good triple in \(U\) induces a
cyclic triangle in \(\mathcal{T}\), and therefore every non-cyclic triangle
\(\{a,b,c\}\) yields a non-planar four-subset among
\(\{u,a,b,c\}\) and \(\{v,a,b,c\}\). These four-subsets considered are pairwise
distinct: those coming from edges contain both \(u\) and \(v\) with the other two elements given by the edge, those
coming from non-cyclic triangles contain exactly one of \(u\) or \(v\) with the other three elements given by the triangle. Writing \(\beta(\mathcal T)\) for the number
of pairs \(a<b\) with \(b\to a\), and \(\tau(\mathcal T)\) for the
number of non-cyclic triangles, we have
\begin{equation}\label{eq:new-np-tournament}
\np(R)-\np(R[U])\geq \beta(\mathcal T)+\tau(\mathcal T).
\end{equation}

It remains to bound the right-hand side.  Let \(d_a\) be the
outdegree of \(a\) and \(\beta_a\) the number of its outneighbours among
\(\{1,\ldots,a-1\}\).  Every non-cyclic triangle has a unique source, and
any two outneighbours of an element form one with that element as
source, so \(\tau(\mathcal T)=\sum_a\binom{d_a}2\) and
\(\beta(\mathcal T)=\sum_a \beta_a\); since \(a\) has only \(m-a\)
successors, \(\beta_a\geq\max\{0,d_a-(m-a)\}\).  Defining
\(f_a(d)=\binom d2+\max\{0,d-(m-a)\}\) we obtain
\begin{equation}\label{eq:tournament-degree-bound}
\beta(\mathcal T)+\tau(\mathcal T)\geq\sum_{a=1}^m f_a(d_a),
\qquad\text{where }\ \sum_{a=1}^m d_a=\binom m2 .
\end{equation}

To bound the right-hand side of
\eqref{eq:tournament-degree-bound} from below, define
\(\Delta_a(d)=f_a(d)-f_a(d-1)\) for \(1\leq d\leq m-1\).
Let us arrange these numbers into an \(m\)-by-\((m-1)\) matrix \(\Delta\) whose \((a,d)\)-entry is \(\Delta_a(d)\).
Since \(f_a(0)=0\), we have \(f_a(d_a)=\sum_{d=1}^{d_a}\Delta_a(d)\). Then the right hand side of \eqref{eq:tournament-degree-bound} is a sum of \(\binom{m}{2}\) entries of the matrix \(\Delta\), no two from the same position. We can check by definition that \(\Delta_a(d) = d-1\) for \(d \leq m-a\) and \(\Delta_a(d) = d\) for \(d > m-a\). Hence, the \(a\)-th row of \(\Delta\) has entries \(0,1,\ldots,m-1\) with \(m-a\) removed, hence each integer \(0,1,\ldots,m-1\) appears exactly \(m-1\) times in \(\Delta\). Taking the \(\binom{m}{2}\) smallest entries from \(\Delta\) and summing them up, we get a lower bound \begin{equation}\label{eq:telescoping-lower-bound}
    \sum_{a=1}^m f_a(d_a) \geq \begin{cases}
       (2q-1)q(q-1)/2 & \text{if \(m=2q\),}\\
       q^3 & \text{if \(m=2q+1\).}
    \end{cases}
\end{equation} Using \eqref{eq:hill-number}, we can check that the right hand side of \eqref{eq:telescoping-lower-bound} is exactly \(H(m+2) - H(m)\). Since \(m = n-2\), combining this with \eqref{eq:new-np-tournament} and \eqref{eq:tournament-degree-bound} gives \(\np(R) \geq \np(R[U]) + H(n) - H(n-2)\), as required.
\end{proof}

In view of Theorem~\ref{thm:antipodal}, one might hope that any
rotation system could be modified somehow into one with an antipodal pair
without increasing the number of non-planar four-subsets. A natural question in this direction is what such modifications should be. It is tempting to take local moves such as transposing two adjacent entries or reversing one cyclic order. After some computational experiments, we have found the following rotation system of order seven:
\[\begin{pmatrix}
2&3&4&5&6&7\\
1&7&4&6&5&3\\
1&2&5&7&6&4\\
1&3&6&2&7&5\\
1&4&7&3&2&6\\
1&5&2&4&3&7\\
1&6&3&5&4&2
\end{pmatrix}.\] For this rotation system, transposing any adjacent entries up to four times or reversing any single row would increase the number of non-planar four-subsets, or trivially get back to the same rotation system. This suggests that a general modification rule may be complicated.

\end{document}